\documentclass[journal]{IEEEtran}
\usepackage{amsfonts}
\usepackage{dsfont}
\usepackage{amssymb}
\usepackage{graphicx}
\usepackage{subfigure}
\usepackage{enumerate}
\usepackage{amsmath}
\usepackage{amsthm}
\usepackage{amsmath}

\usepackage[bookmarks=false]{hyperref}
\usepackage{algorithm}
\usepackage{algpseudocode}
\usepackage{cite}
\hypersetup{hidelinks}

\usepackage[utf8]{inputenc}

\newcommand{\bp}{\begin{proof} \small }
\newcommand{\ep}{\end{proof} \normalsize}
\newcommand{\epx}{\end{proof} \small}
\newcommand{\bpa}{\begin{proofappx} \footnotesize }
\newcommand{\epa}{\end{proofappx} \small }
\newtheorem{theorem}{Theorem}

\newtheorem{lemma}{Lemma}

\newtheorem{remark}{Remark}
\newtheorem{conjecture}{Conjecture}
\newtheorem*{theorem*}{Theorem}
\newtheorem*{proposition*}{Proposition}
\newtheorem*{corollary*}{Corollary}
\newtheorem*{lemma*}{Lemma}
\newtheorem*{assumption*}{Assumption}
\newtheorem*{definition*}{Definition}
\newtheorem*{claim*}{Claim}

\newcommand{\be}{\begin{equation}}
\newcommand{\ee}{\end{equation}}
\newcommand{\bs}{\begin{subequations}}
\newcommand{\es}{\end{subequations}}
\newcommand{\bq}{\begin{eqnarray}}
\newcommand{\eq}{\end{eqnarray}}
\newcommand{\bqn}{\begin{eqnarray*}}
\newcommand{\eqn}{\end{eqnarray*}}

\newcommand{\ba}{\left[ \begin{array}}
\newcommand{\ea}{\\ \end{array} \right]}
\newcommand{\ben}{\begin{enumerate}}
\newcommand{\een}{\end{enumerate}}

\def\real{{\mathchoice%
{\hbox{\rm\setbox1=\hbox{I}\copy1\kern-.45\wd1 R}}
{\hbox{\rm\setbox1=\hbox{I}\copy1\kern-.45\wd1 R}}
{\hbox{\scriptsize\rm\setbox1=\hbox{I}\copy1\kern-.45\wd1 R}}
{\hbox{\scriptsize\rm\setbox1=\hbox{I}\copy1\kern-.45\wd1 R}}}}

\def\Zint{{\mathchoice{\setbox1=\hbox{\sf Z}\copy1\kern-.75\wd1\box1}
{\setbox1=\hbox{\sf Z}\copy1\kern-.75\wd1\box1}
{\setbox1=\hbox{\scriptsize\sf Z}\copy1\kern-.75\wd1\box1}
{\setbox1=\hbox{\scriptsize\sf Z}\copy1\kern-.75\wd1\box1}}}
\newcommand{\complex}{ \hbox{\rm C\kern-0.45em\rule[.07em]{.02em}{.58em}%
\kern 0.43em}}

\begin{document}
	%
	\title{Adaptive Learning-Based Task Offloading for Vehicular Edge Computing Systems}


	\author{Yuxuan~Sun,~\IEEEmembership{Student Member,~IEEE,}
		Xueying Guo,~\IEEEmembership{Member,~IEEE,}		
		Jinhui Song,
		Sheng~Zhou,~\IEEEmembership{Member,~IEEE,}
		Zhiyuan Jiang,~\IEEEmembership{Member,~IEEE,}
		Xin Liu,~\IEEEmembership{Fellow,~IEEE,}
		and Zhisheng Niu,~\IEEEmembership{Fellow,~IEEE} 
		\thanks{Y. Sun, J. Song, S. Zhou,  Z. Jiang and Z. Niu are with Beijing National Research Center for Information Science and Technology, Department of Electronic Engineering, Tsinghua University, China. Emails: \{sunyx15, sjh14\}@mails.tsinghua.edu.cn, \{sheng.zhou, zhiyuan, niuzhs\}@tsinghua.edu.cn.}  
		\thanks{X. Guo and X. Liu are with the Department of Computer Science, University of California, Davis, CA, USA. Emails: guoxueying@outlook.com, xinliu@ucdavis.edu.}
		\thanks{This work is sponsored in part by the Nature Science Foundation of China (No. 61871254, No. 91638204, No. 61571265, No. 61861136003, No. 61621091), National Key R\&D Program of China 2018YFB0105005, NSF through grants CNS-1547461, CNS-1718901, IIS-1838207, and Intel Collaborative Research Institute for Intelligent and Automated Connected Vehicles.
		(Corresponding author: Sheng Zhou.)}
		\thanks{Part of this work has been published in IEEE ICC 2018 \cite{Sun2018ICC}.}
	}

	\maketitle

	\begin{abstract}
		The vehicular edge computing (VEC) system integrates the computing resources of vehicles, and provides computing services for other vehicles and pedestrians with task offloading.
		However, the vehicular task offloading environment is dynamic and uncertain, with fast varying network topologies, wireless channel states and computing workloads.
		These uncertainties bring extra challenges to task offloading.
		In this work, we consider the task offloading among vehicles, and propose a solution that enables vehicles to learn the offloading delay performance of their neighboring vehicles while offloading computation tasks.	
		We design an adaptive learning-based task offloading (ALTO) algorithm based on the multi-armed bandit (MAB) theory, in order to minimize the average offloading delay.
		ALTO works in a distributed manner without requiring frequent state exchange, and is augmented with input-awareness and occurrence-awareness to adapt to the dynamic environment.
		The proposed algorithm is proved to have a sublinear learning regret.
		Extensive simulations are carried out under both synthetic scenario and realistic highway scenario, and results illustrate that the proposed algorithm achieves low delay performance, and decreases		
		the average delay up to $30\%$ compared with the existing upper confidence bound based learning algorithm.
		
	\end{abstract}
	

	%
	\IEEEpeerreviewmaketitle
	
	\section{Introduction}
	By deploying computing resources at the edge of the network, mobile edge computing (MEC) can provide low-latency, high-reliability computing services for mobile devices \cite{hu2015mobile, shih17}.
	A major problem in MEC is how to perform \emph{task offloading}, i.e., whether or not to offload each task, and how to manage radio and computing resources to execute tasks, which has been widely investigated recently, see surveys \cite{mao2017mobile, mach2017mobile, yu2018survey} and technical papers \cite{you2016energy, chen2016efficient, jin2018}.
	
	  
	To support autonomous driving and a vast variety of on-board infotainment services, vehicles are equipped with substantial computing and storage resources. It is forecast that each self-driving car will have computing power of $10^6$ dhrystone million instructions executed per second (DMIPS) in the near future\cite{intel}, which is tens of times that of the current laptops.
	Vehicles and infrastructures like road side units (RSUs) can contribute their computing resources to the network. This forms the Vehicular Edge Computing (VEC) system \cite{abdel2015vehicle,bitam2015vanet,choo2017sdvc}, that can process computation tasks from vehicular driving systems, on-board mobile devices and pedestrians for various applications.

	In this paper, we focus on the task offloading among vehicles, i.e., the driving systems or passengers of some vehicles generate computation tasks, while some other surrounding vehicles can provide computing services.
	We call the vehicles that require task offloading \emph{task vehicles (TaVs)}, and vehicles who can help to execute tasks \emph{service vehicles (SeVs)}.
	We design a distributed task offloading algorithm to minimize the average delay, where the task offloading decision is made by each TaV individually.

	Multiple SeVs might be available to process each task, and a key challenge is the lack of accurate state information of SeVs in the dynamic VEC environment. The network topology and the wireless channel states vary rapidly due to the movements of vehicles \cite{aibo}, and the computation workloads of SeVs fluctuate across time. These factors are difficult to model or to predict, so that the TaV has no idea \emph{in prior} which SeV performs the best in terms of delay performance.
	
	Our solution is \emph{learning while offloading}, i.e., the TaV is able to learn the delay performance while offloading tasks. To be specific, we adopt the multi-armed bandit (MAB) framework to design our task offloading algorithm \cite{auer2002finite}. The classical MAB problem aims at balancing the exploration and exploitation tradeoff in the learning process: to explore different candidate actions that lead to good estimates of their reward distributions, while to exploit the learned information to select the empirically optimal actions. The upper confidence bound (UCB) based algorithms, such as UCB1 and UCB2, have been proposed with strong performance guarantee \cite{auer2002finite}, and applied to the wireless networks to learn the unknown environments \cite{chen2011opp, shen2016non,sun2017emm}.
	
	However, in our task offloading problem, the movements of vehicles lead to a dynamic candidate SeV set, and the workload of each task is time-varying, leading to a varying cost in exploring the suboptimal actions. These factors have not been addressed by existing MAB schemes, which motivates us to specifically adapt the MAB framework in the vehicular task offloading scenario.	
	Our key contributions include:
	
	1) We propose an adaptive learning-based task offloading (ALTO) algorithm based on MAB theory, in order to guide the task offloading of TaVs and minimize the average offloading delay.	
	ALTO algorithm works in a distributed manner and enables the TaV to learn the delay performance of candidate SeVs while offloading tasks. The proposed algorithm is of low computational complexity, and does not require the exchange of accurate state information like channel states and computing workloads between vehicles, so that it is easy to implement in the real VEC system.
	
	2) Two kinds of \emph{adaptivity} are augmented with the proposed ALTO algorithm: \emph{input-awareness} and \emph{occurrence-awareness}, by adjusting the exploration weight according to the workloads of tasks and the appearance time of SeVs.
	Different from our previous theoretical work \cite{wu2017adaptive} which only considers time-varying workloads of tasks with fixed actions, we consider a more general case with dynamic candidate SeVs (actions), and prove that ALTO can effectively balance the exploration and exploitation in the dynamic vehicular environment with sublinear learning regret.
	%
	
	3) Extensive simulations are carried out under a synthetic scenario, as well as a realistic highway scenario using system level simulator Veins. Results illustrate that our proposed algorithm can achieve low delay performance, and provide guidelines for the settings of key design parameters.
	

	The rest of this paper is organized as follows. We introduce the related work in Section \ref{relatedwork}. The system model and problem formulation is introduced in Section \ref{sys}, and the ALTO algorithm is then proposed in Section \ref{algo}. The learning regret is analyzed in Section \ref{per}. Simulation results are then provided in Section \ref{sim}, and finally comes the conclusions in Section \ref{con}.

\section{Related Work} \label{relatedwork}

\subsection{VEC Architecture and Use Cases}
\begin{figure*}  [!t]
	\centering
	\includegraphics[width=0.85\textwidth]{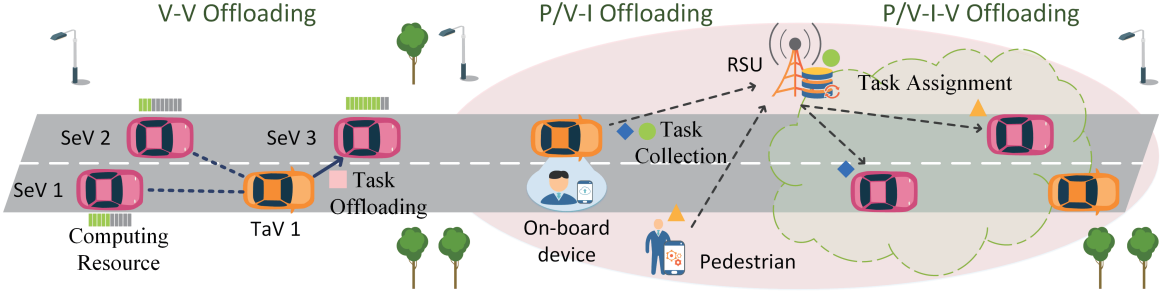}\\  
	\caption{An illustration of the VEC architecture and three major offloading modes.}\label{system}
\end{figure*}

An illustration of the VEC architecture is shown in Fig. \ref{system}. The development of vehicle-to-everything (V2X) communication techniques enable vehicle-to-vehicle (V2V), vehicle-to-infrastructure (V2I) and vehicle-to-pedestrian (V2P) communications, so that tasks can be offloaded to other vehicles through different kinds of routes. Specifically, there are three major offloading modes:

\begin{itemize}		
	\item \emph{Vehicle-Vehicle (V-V) Offloading}: Vehicles directly offload tasks to their surrounding vehicles with surplus computing resources in a distributed manner. In this case, each individual vehicle may not be able to acquire the global state information for task offloading decisions, and there might be no coordinations for task scheduling.
	
	\item \emph{Pedestrian/Vehicle-Infrastructure-Vehicle (P/V-I-V) Offloading}: When there are no other neighboring vehicles for task offloading,
	one solution is that tasks are first offloaded to the infrastructures alongside, and then assigned to other vehicles in a centralized manner.
	
	
	\item \emph{Pedestrian/Vehicle-Infrastructure (P/V-I) Offloading}: In this mode, tasks are offloaded to the infrastructures for direct processing.
	
\end{itemize}

	Similar to the traditional cloud computing services, the VEC system can provide infrastructure as a service (IaaS), platform as a service (PaaS) and software as a service (SaaS) \cite{choo2017sdvc}, and support a wide variety of applications.
For example, cooperative collision avoidance and collective environment perception are necessary for safety driving, where sensing data is generated by a group of vehicles and processed by some of them\cite{5gv2x,zhangshan}. 
In vehicular crowd sensing, the video recordings and images are generated by vehicles and required to be analyzed in real time, in order to supervise the traffic, monitor the road conditions and navigate car parkings \cite{vehcrowd}.
The computing resources of vehicles may be underutilized by the aforementioned vehicular applications \cite{abdel2015vehicle}, which can further provide services for entertainments and multimedia applications, such as cloud gaming, virtual reality, augmented reality and video trans-coding \cite{fengave}.

\subsection{Task Offloading Algorithms}
	
	There are some existing efforts investigating the task scheduling and computing resource management problem in VEC.
	A software-defined VEC architecture is proposed in \cite{choo2017sdvc}.	Inspired by the software-defined network, a centralized controller is designed to periodically collect the state information of vehicles, including mobility and resource occupation, and manage radio and computing resources upon task requests. 
	In terms of P/V-I-V offloading, a semi-Markov decision based centralized task assignment problem is formulated in \cite{zheng2015smdp}, in order to minimize the average system cost by jointly considering the delay of tasks and the energy consumption of mobile devices. 
	Ref. \cite{Jiang2017IoT} further introduces task replication technique to improve the service reliability of VEC, where task replicas can be offloaded to multiple vehicles to be processed simultaneously. 
	However, a key drawback of the centralized framework is that, it requires frequent state information update to optimize the system performance, which is of high signaling overhead. 
	
	An alternative method is to make task offloading decisions by the task generators in a distributed manner. An autonomous vehicular edge framework which enables V-V and V-I offloading is proposed in \cite{fengave}, followed by a task scheduling algorithm based on ant colony optimization. However, when the number of vehicles is large, the computational complexity can be quite high.
	We will design a distributed task offloading algorithm with low complexity.

	\section{System Model and Problem Formulation} \label{sys}
	


	\subsection{V-V Offloading: System Overview}
	We consider V-V offloading in the VEC system, where vehicles involved in the task offloading are classified into two categories: \emph{TaVs} are the vehicles that generate and offload computation tasks for cloud execution, while \emph{SeVs} are the vehicles with sufficient computing resources that can provide computing services. Note that the role of each vehicle depends on the sufficiency of its computing resources, and is not fixed to TaV or SeV during the trip.
			
	TaVs can offload tasks to their neighboring SeVs. Each TaV may have multiple candidate SeVs that can process the tasks, and each task is offloaded to a single SeV and executed by it. 
	As shown in Fig. \ref{system}, for TaV 1, there are 3 candidate SeVs (SeV 1-3), and currently the task is offloaded to SeV 3.
	
	In this work, we design distributed task offloading algorithm to minimize the delay performance, by letting each TaV decide which SeV should serve each task independently, without inter-TaV cooperations. Moreover, we do not make any assumptions on the service disciplines of SeVs, nor the mobility models of vehicles.

	\subsection{Task Offloading Procedure}
	Since offloading decisions are made in a distributed manner, we then focus on a single TaV of interest and model the task offloading problem. Consider a discrete-time VEC system. 
	There are four procedures for task offloading within each time period:
	
	\textbf{SeV discovery:} 
	The TaV discovers neighboring SeVs within its communication range, and selects those in the same moving direction as candidates. Here the driving states of each vehicle, including speed, location and moving direction, can be acquired by other neighboring vehicles through vehicular communication protocols. For example, in dedicated short-range communication (DSRC) standard \cite{kenney2011dsrc}, the periodic beaconing messages can provide these state information. Denote the candidate SeV set in time period $t$ by $\mathcal{N}(t)$, which may change across time since vehicles are moving.  
	And due to the unknown mobility model, candidate SeVs in the future are unknown in prior. Besides, assume that $\mathcal{N}(t) \neq \emptyset$ for $\forall t$, otherwise the TaV can seek help from RSUs along the road, which is beyond the scope of this paper.   
	
	\textbf{Task upload:}
	After updating the candidate SeV set $\mathcal{N}(t)$ at the beginning of each time period, the TaV selects one SeV $n \in \mathcal{N}(t)$ and uploads the computation task. 
	Denote the input data size of the task generated in time period $t$ by $x_t$ (in bits), which is required to be transmitted from TaV to SeV.
	The uplink wireless channel state between TaV and SeV $n \in \mathcal{N}(t)$ is denoted by $h^{(u)}_{t,n}$, and the interference power at SeV $n$ is $I^{(u)}_{t,n}$. 
	We assume that the wireless channel state remains static during the uploading process of each computation task.
	Given the fixed transmission power $P$, channel bandwidth $W$ and noise power $\sigma^2$, the uplink transmission rate $r^{(u)}_{t,n}$ between the TaV and SeV $n$ is   
	\begin{align} \label{uplink_rate}
	r^{(u)}_{t,n} = W\log_2\left(1 + \frac{Ph^{(u)}_{t,n}}{\sigma^2 + I^{(u)}_{t,n}}\right).
	\end{align}
	And the transmission delay $d_\mathrm{up}(t,n)$ of uploading the task to SeV $n$ in time period $t$ is given by      
	\begin{align}  \label{uplink}
	d_\mathrm{up}(t,n) = \frac{x_t }{r^{(u)}_{t,n}}.
	\end{align}

	\textbf{Task execution:} 
	The selected SeV $n$ processes the task after receiving the input data from the TaV. For the task generated in time period $t$, the total workload is given by $x_tw_t$, where $w_t$ is computation intensity (in CPU cycles per bit) representing how many CPU cycles are required to process one bit input data \cite{mao2017mobile}. The computation intensity $w_t$ of the task mainly depends on the nature of applications. 
	
	The computing capability of SeV $n$ is described by its maximum CPU frequency $F_n$ (in CPU cycles per bit), and the allocated CPU frequency to the task of TaV in time period $t$ is denoted by $f_{t,n}$. The SeV may deal with multiple computation tasks simultaneously, and adopt dynamic frequency and voltage scaling (DVFS) technique to dynamically adjust the CPU frequency \cite{zhang2013energy}, and thus we have $f_{t,n}\in [0,F_n]$. 
	We assume that $f_{t,n}$ remains static during each time period $t$, and  
		each computation task can be completed within each time period due to the timely requirements. Tasks of larger workloads can be further partitioned into multiple subtasks \cite{sun2017emm, grundmann2010efficient}, so that each subtask is offloaded to and processed by a SeV within one time period.
	Then the computation delay can be written as	
	\begin{align}
	d_\mathrm{com}(t,n) = \frac{x_t w_t}{f_{t,n}}.
	\end{align}

	\textbf{Result feedback:} 
	Upon the completion of task execution, the selected SeV $n$ transmits back the result to the TaV. 
	Let $h^{(d)}_{t,n}$ denote the downlink wireless channel state, which is assumed to be static during the transmission of each result. The interference at the TaV is denoted by $I^{(d)}_{t}$.
	Similar to \eqref{uplink}, the downlink transmission rate $r^{(d)}_{t,n}$ from SeV $n$ to TaV can be written as
	\begin{align}
	r^{(d)}_{t,n} = W\log_2\left(1 + \frac{Ph^{(d)}_{t,n}}{\sigma^2 + I^{(d)}_{t}}\right).
	\end{align}
	The data volume of the computation result in time period $t$ is denoted by $y_t$ (in bits), and thus the downlink transmission delay from SeV $n$ to the TaV is
	\begin{align}  
	d_\mathrm{dow}(t,n)  = \frac{y_t }{r^{(d)}_{t,n}}.
	\end{align}
	
	Then the sum delay $d_\mathrm{sum}(t,n)$ of offloading the task to SeV $n$ in time period $t$ can be given by
	\begin{align}  
	d_\mathrm{sum}(t,n) = d_\mathrm{up}(t,n)+d_\mathrm{com}(t,n) +d_\mathrm{dow}(t,n).
	\end{align}

	\subsection{Problem Formulation} \label{pro}
	Consider a total number of $T$ time periods. Our objective is to minimize the average offloading delay, by guiding the task offloading decisions of the TaV on which SeV should serve each task. The task offloading problem is formulated as 
	\begin{align} \label{obj}
	\textbf{P1:}~\min_{a_1,...,a_T} \frac{1}{T}\sum_{t=1}^{T}d_\mathrm{sum}(t,a_t),  
	\end{align}	
	where $a_t$ is the optimization variable, which represents the index of SeV selected in time period $t$, with $a_t\in\mathcal{N}(t)$.
	
	\textbf{Availability of state information:}
	The state information related to the delay performance can be classified into two categories based on its ownership: parameters of each task, including the input and output data volumes $x_t$, $y_t$ and computation intensity $w_t$, are known by the TaV upon the generation of  each task. The uplink and downlink transmission rates $r^{(u)}_{t,n}$, $r^{(d)}_{t,n}$ and the allocated CPU frequency $f_{t,n}$ are closely related to the SeV. 
	If all these states are exactly known by the TaV before offloading each task, the sum delay $d_\mathrm{sum}(t,n) $ of SeV $n\in\mathcal{N}(t)$ can then be calculated, and the optimization problem \textbf{P1} is easy to solve with
	\begin{align}
		a_t=\min_{n \in \mathcal{N}_t}d_\mathrm{sum}(t,n).
	\end{align}
	
	However, due to the mobility of vehicles, the transmission rates vary fast across and are difficult to predict. Since there is no cooperation between TaVs, the computation loads at SeVs dynamically change, making the allocated CPU frequency vary across time. Moreover, exchanging these state information between the TaV and all candidate SeVs causes high signaling overhead.	
	Therefore, the TaV may lack the state information of SeVs, and can not realize which SeV provides the lowest delay when making offloading decisions.	
	
	\textbf{Learning while offloading:}
	To overcome the unavailability of the state information of SeVs, we propose the approach \emph{learning while offloading}: the TaV can observe and learn the delay performance of candidate SeVs while offloading computation tasks. Specifically,  the SeV $a_{t}$ in time period $t$ is selected according to the historical delay observations $d(1,a_1), d(2,a_2), ...,d(t-1,a_{t-1})$, without acquiring the exact transmission rates and CPU frequency.
	We aim to design a learning algorithm that minimizes the expectation of offloading delay, written as
	\begin{align} \label{obj2}
	\textbf{P2:}~\min_{a_1,...,a_T}\frac{1}{T}\mathbb{E}\left[\sum_{t=1}^{T}d_\mathrm{sum}(t,a_t)\right].
	\end{align}	
	
	In the rest of the paper, we consider a simplified version of \textbf{P2} by assuming that the input data size $x_t$ of task is time-varying, but the computation intensity $w_t$ and the ratio of output and input data volume $y_t/x_t$ remains constant across time. In practical, this is a valid assumption when tasks are generated by the same type of application. 
	Let $y_t/x_t=\alpha_0$ and $w_t=\omega_0$ for $\forall t$. Then the sum delay of offloading the task to SeV $n$ in time period $t$ can be transformed as
	\begin{align}
		d_\mathrm{sum}(t,n)=x_t\left(\frac{1 }{r^{(u)}_{t,n}}+\frac{\alpha_0}{r^{(d)}_{t,n}}+ \frac{\omega_0}{f_{t,n}}\right).
	\end{align}
	Define the \emph{bit offloading delay} as
	\begin{align}
		u(t,n)=\frac{1 }{r^{(u)}_{t,n}}+\frac{\alpha_0}{r^{(d)}_{t,n}}+ \frac{\omega_0}{f_{t,n}},
	\end{align}
	which represents the sum delay of offloading one bit input data of the task to SeV $n$ in time period $t$. The bit offloading delay $u(t,n)$ reflects the service capability of each candidate SeV, which is what the TaV needs to learn.
	
	Finally, the optimization problem can be written as
	\begin{align} \label{obj3}
	\textbf{P3:}~\min_{a_1,...,a_T}\frac{1}{T}\mathbb{E}\left[\sum_{t=1}^{T}x_tu(t,n)\right].
	\end{align}

%
	
	\section{Adaptive Learning-Based Task Offloading Algorithm}\label{algo}
	In this section, we develop a learning-based task offloading algorithm based on MAB, which enables the TaV to learn the delay performance of candidate SeVs and minimizes the expected offloading delay.
	
	Our task offloading problem \textbf{P3} requires online sequential decision making, which can be solved according to the MAB theory. Each SeV corresponds to an arm whose loss (bit offloading delay) is governed by an unknown distribution. The TaV is the decision maker who tries an arm at a time and learns the estimation of its loss, in order to minimize the expectation of cumulative loss across time.	
	However, the variations of input data size $x_t$ and candidate SeV set $\mathcal{N}_t$  incapacitate existing algorithms of MAB, such as UCB1 and UCB2, in the VEC system. 
	
	In this work, we propose an Adaptive Learning-based Task Offloading (ALTO) algorithm which is aware of both the input data size of tasks and the occurrence of vehicles, as shown in Algorithm 1.
	Parameter $\beta$ is a constant weight, and $k_{t,n}$ records the number of tasks that have been offloaded to SeV $n$ up till time $t$. The occurrence time of SeV $n$ is recorded by $t_n$, and the input data size $x_t$ is normalized to be $\tilde{x}_t$ within $[0,1]$ as:
	\begin{align} \label{normailized_x}
	\tilde{x}_t=\max \left\{ \min \left( \frac{x_t-x^-}{x^+-x^-},1\right), 0\right\},
	\end{align}
	where $x^+$ and $x^-$ are the upper and lower thresholds to normalize $x_t$. In particular, if $x^+=x^-$, $\tilde{x}_t=0$ when $x_t\leq x^-$, and $\tilde{x}_t=1$ when $x_t > x^-$.
	
	In Algorithm 1, Lines 3-5 are the initialization phase, which is called whenever new SeVs occur as candidates. The TaV selects the newly appeared SeV $n$ once and offloads the task, in order to get an initial estimation of its bit offloading delay.
	
	Lines 7-12 are the main loop of the learning process, inspired by the volatile UCB (VUCB) algorithm \cite{bnaya2013social} and the our previous work on opportunistic MAB \cite{wu2017adaptive}. During each time period, the TaV gets the data volume $x_t$ before offloading the task and calculates $\tilde{x}_t$. The utility function defined in \eqref{ada_utility} is used to evaluate the service capability of each SeV, which consists of the empirical bit offloading delay $\bar u_{t,n}$ and a padding function. Specifically, $\bar u_{t,n}$ is the average bit offloading delay of SeV $n$ observed until time period $t$. And the padding function jointly considers the input data size and occurrence time of each SeV, in order to balance the exploration and exploitation in the learning process, and adapt to the dynamic VEC environment.
	The offloading decision is then made according to \eqref{ada_obj}, by selecting the SeV with minimum utility. Finally, the offloading delay is observed upon result feedback, and $\bar{u}_{t,a_t}$ and $k_{t,a_t}$ is updated. 
	
		\begin{algorithm}
			\caption{ALTO: Adaptive Learning-based Task Offloading Algorithm}
			\begin{algorithmic}[1]
				\State \textbf{Input}: parameters $\alpha_0$, $\omega_0$, $\beta$, $x^+$ and $x^-$.
				\For {$t=1,...,T$}
				\If { Any SeV $n \in \mathcal{N}(t)$ has not connected to TaV}
				\State Connect to SeV $n$ once.
				\State Update $\bar u_{t,n}=d_\mathrm{sum}(t,n)/x_t$, $k_{t,n}=1$, $t_n=t$.
				\Else
				\State Observe $x_t$, calculate $\tilde{x}_t$.
				\State Calculate the utility function of each candidate SeV $n \in \mathcal{N}(t)$:
				\begin{align} \label{ada_utility}
				\hat{u}_{t,n}=\bar{u}_{t-1,n}-\sqrt{\frac{\beta(1-\tilde{x}_t)\ln (t-t_n)}{k_{t-1,n}}}.
				\end{align}		
				\State Offload the task to SeV $a_t$, such that:
				\begin{align} \label{ada_obj}
				a_t=\arg\min_{n \in \mathcal{N}(t)} \hat{u}_{t,n}.
				\end{align} 
				\State Observe the sum offloading delay $d_\mathrm{sum}(t,a_t)$.
				\State Update $\bar{u}_{t,a_t}\leftarrow \frac{ \bar{u}_{t-1,a_t}k_{t-1,a_t}+d_\mathrm{sum}(t,a_t)/x_t}{k_{t-1,a_t}+1} $.
				\State Update $k_{t,a_t}\leftarrow k_{t-1,a_t}+1$.
				\EndIf
				\EndFor		
			\end{algorithmic}
		\end{algorithm}
	
	Two kinds of adaptivity of the algorithm are highlighted as follows.

	\textbf{Input-awareness:}
	The input data size $x_t$ can be regarded as a weight factor on the offloading delay. Intuitively, when $x_t$ is small, even if the TaV selects a poorly performed SeV, the sum offloading delay will not be too large. On the other hand, when $x_t$ is large, selecting a SeV with weak service capability brings great delay degradation. Therefore, the padding function is proportional to $\sqrt{1-\tilde{x}_t}$ that is non-increasing as $x_t$ grows, so that ALTO explores more when $x_t$ is small, while exploits more when $x_t$ is large.
	
	\textbf{Occurrence-awareness:}
	The random presences of SeVs are also considered, and the proposed ALTO algorithm has occurrence-awareness. To be specific, for any newly appeared SeV,  $\sqrt{\frac{\ln (t-t_n)}{k_{t-1,n}}}$ is large due to the small number of selections $k_{t-1,n}$, so that ALTO tends to explore more. Meanwhile, ALTO is able to exploit the learned information of any existing SeV, since more times of connections lead to a small value of the padding function.

	\subsection{Complexity}
	 
	In our proposed ALTO algorithm, the computational complexity of calculating the utility functions of all candidate SeVs in Line 8 is $O(N)$, where $N=|\mathcal{N}(t)|$ is the number of candidate SeVs in time period $t$. The task offloading decision made in Line 9 is a minimum seeking problem, with complexity $O(N)$. Updating the empirical bit offloading delay $\bar u_{t,a_t}$  and offloaded times $k_{t,a_t}$ has a complexity of $O(1)$. Therefore, within each time period, the total computational complexity of running ALTO to offload one task is $O(N)$.
	Assume that there are totally $M$ tasks required to be offloaded in the VEC system. Since TaVs offload tasks independently, the total amount of computation is $O(MN)$.
	
	An ant colony optimization based distributed task offloading algorithm is proposed in \cite{fengave}. According to Section V.D, the computational complexity is $O(KM^2N)$, where $K$ is the number of iterations required by the ant colony optimization. Therefore, ALTO is of lower complexity than the existing algorithm in \cite{fengave}.
	
	\subsection{Signaling Overhead}

		Considering the distributed V-V offloading case, the complete-state task offloading (CSTO) policy is that, the TaV obtains the accurate state information of all candidate SeVs, evaluates their delay performance, and selects the SeV with minimum offloading delay.
		Compared with the CSTO policy, our proposed ALTO algorithm is of lower signaling overhead and much easier to implement in the real VEC system.

		First, the uplink and downlink wireless channel states, allocated CPU frequency and interference of each candidate SeV are not required to know by the ALTO algorithm. Therefore, for each TaV, offloading a task can save at least $N$ signaling messages for the state information of the $N$ candidate SeVs, and $MN$ signaling messages can be saved for $M$ tasks.
		Second, when a SeV is serving multiple TaVs simultaneously, the CSTO policy needs to know the task workload of TaVs to allocate computing resources of the SeV. In this case, more signaling messages are generated by the CSTO policy.
		Last but not least, frequent signaling exchange may lead to additional collisions and retransmissions, and the delayed state information may not be accurate.
		The proposed ALTO algorithm enables each TaV to learn the state information of SeVs instead of obtaining them from signaling messages, and thus reduces the signaling overhead.

	\section{Performance Analysis}\label{per}
	In this section, we characterize the delay performance of the proposed ALTO algorithm.
	We adopt the \emph{learning regret} of delay as the performance criteria, which is widely used in the MAB theory. 
	Compared with the existing UCB based algorithms in \cite{auer2002finite}, two major modifications in ALTO are the occurrence time $t_n$ and normalized input $\tilde{x}_t$. We first evaluate their impacts on the learning regret separately, and then jointly analyze these two factors.
	
	\subsection{Definition of Learning Regret}
	Define an \emph{epoch} as the interval during which candidate SeVs remain identical.
	The total number of epochs during the considered $T$ time periods is denoted by $B$, and let $\mathcal{N}_b$ be the candidate SeV set of the $b$th epoch, where $b=1,2,...,B$.
	Let $t_b$ and $t'_b$ be the start and end time of the $b$th epoch, with $t_1=1$ and $t'_B=T$.
	
	For theoretical analysis, we assume that for each SeV $n$, its bit offloading delay $u(t,n)$ is i.i.d. over time and independent of others. We will show in Section \ref{sim} through simulation results that without this assumption, ALTO still works well.
	
	Define the mean bit offloading delay of each candidate SeV $n$ as $\mu_n=\mathbb{E}_t[u(t,n)]$. During each epoch, let $\mu_b^*=\min_{n\in\mathcal{N}_b}\mu_n$ be the optimal bit offloading delay, and $a_b^*=\arg\min_{n\in\mathcal{N}_b}\mu_n$ the index of the optimal SeV. Note that $\mu_b^*$ and $a_b^* $ are unknown in prior.
	
	The learning regret represents the expected cumulative performance loss of sum offloading delay brought by the learning process, which is compared with the genie-aided optimal policy where the TaV always selects the SeV with maximum service capability. The learning regret by time period $T$ can be written as
	\begin{align}
	R_T = \sum_{b=1}^{B} \mathbb{E}\left[\sum_{t=t_b}^{t'_b}x_t\left(u(t,n)-\mu_b^*\right)\right],	
	\end{align}
	
	 In the following subsections, we will characterize the upper regret bound of ALTO algorithm.
	
	\subsection{Regret Analysis under Dynamic SeV Set and Identical Input}	
	We first assume that the input data size is not time-varying, and analyze the learning regret under varying SeV set. Let $x_t=x_0$ for $\forall t$, and $x^+=x^-=x_0$, then $\tilde{x}_t=0$. The utility function \eqref{ada_utility} becomes
	\begin{align}
		\hat{u}_{t,n}=\bar{u}_{t-1,n}-\sqrt{\frac{\beta\ln (t-t_n)}{k_{t-1,n}}},
	\end{align}
	and the learning regret
	\begin{align}
		R_T =x_0 \sum_{b=1}^{B} \mathbb{E}\left[\sum_{t=t_b}^{t'_b}\left(u(t,n)-\mu_b^*\right)\right].
	\end{align}
	
	Also, define the maximum bit offloading delay during the $T$ time periods as $u_m=\sup_{t,n} u(t,n) $, the performance difference between any suboptimal SeV $n\in\mathcal{N}_b$ and the optimal SeV in the $b$th epoch $\delta_{n,b}=(\mu_n-\mu_b^*)/u_m$. Let $\beta=\beta_0u_m^2$, where $\beta_0$ is a constant.
	
	The learning regret within each epoch is upper bounded in Lemma \ref{vucb_epoch}.
	
	\begin{lemma} \label{vucb_epoch}
		Let $\beta_0=2$, the learning regret of ALTO with dynamic SeV set and identical input data size has an upper bound in each epoch. Specifically, in the $b$th epoch:
		\begin{align}
		R_{b}\leq x_0u_m \left[\sum_{n\neq a_b^*}\frac{8\ln (t'_b-t_n)}{\delta_{n,b}} + \left(1+ \frac{\pi^2}{3}\right)\sum_{n\neq a_b^*}\delta_{n,b} \right].
		\end{align}
	\end{lemma}
	\begin{proof}
		See Appendix \ref{a1}.
	\end{proof}
	
	Then we have the following Theorem \ref{vucb_T} that provides the upper bound of the learning regret over $T$ time periods.
	
	\begin{theorem} \label{vucb_T}
		Let $\beta_0=2$. For a given time horizon $T$, the total learning regret $R_T$ of ALTO dynamic SeV set and identical input data size has an upper bound as follows:
		\begin{align}
		R_{T}\leq x_0u_m\sum_{b=1}^{B} \left[\sum_{n\neq a_b^*}\frac{8\ln T}{\delta_{n,b}} + O(1) \right].
		\end{align}
	\end{theorem}
	\begin{proof}
		See Appendix \ref{a3}.
	\end{proof}
	
	Theorem \ref{vucb_T} implies that, our proposed ALTO algorithm provides a sublinear learning regret compared to the genie-aided optimal policy. 
	To be specific, within each epoch, the learning regret is governed by $O(\ln T)$, and inversely proportional to the performance difference $\delta_{n,b}$ of optimal SeV and suboptimal SeV $n\neq a_b^*$.
	Moreover, for any finite time horizon $T$ with $B$ epochs, ALTO achieves $O(B\ln T)$ learning regret.
	
\begin{remark}
		The random appearance and disappearance of SeVs affect the number of epochs $B$ and the learning regret $O(B\ln T)$. Within a fixed number of time periods, higher randomness of SeVs results in a more dynamic environment, and thus higher learning regret.
	\end{remark}
		
	\begin{remark}
		To prove Lemma \ref{vucb_epoch} and Theorem \ref{vucb_T}, we have to normalize the bit offloading delay $u(t,n)$ within $[0,1]$ for $\forall t, n$, by setting $u_m=\sup_{t,n} u(t,n) $. In practical, the exact value of $u_m$ is not easy to acquire in prior. Instead, $u_m$ can be set to the maximum $u(t,n)$ that has been observed till the current time period. 
	\end{remark}
	
	\subsection{Regret Analysis under Varying Input and Fixed Candidate SeVs}
	We then characterize the upper bound of the learning regret within a single epoch, and consider that the input data size $x_t$ is random and continuous.
	Let $B=1$. The optimal SeV is $a^*= \arg\min_{n\in\mathcal{N}_1}\mu_n$, and its mean bit offloading delay $\mu^*=\min_{n\in\mathcal{N}_1}\mu_n$.
	The learning regret can be simplified as
	\begin{align}
		R_T =\mathbb{E}\left[\sum_{t=1}^{T}x_t(u(t,n)-\mu^*)\right].
	\end{align}
	
	The following theorem bounds the learning regret under varying input data size and fixed candidate SeV set.
	
	\begin{theorem} \label{ada}
		Let $\beta_0=2$, and $\mathbb{P}\{x_t\leq x^-\}>0$. For any finite time horizon $T$, we have:
		
		(1)	When $x^+\geq	x^-$, the expected number of tasks $k_{T,n}$ offloaded to any SeV $n\neq a^*$ can be bounded as 
		\begin{align}\label{ada1}
		\mathbb{E}[k_{T,n}] \leq\frac{8\ln T }{\delta^2_{n}} +O(1).   
		\end{align}
		
		(2) With $x^+=	x^-$, the learning regret can be bounded as
		\begin{align}\label{ada2}
		R_T \leq u_m\sum_{n\neq a^*}\left[ \frac{8\ln T \mathbb{E}[x_t|x_t\leq x^-]}{\delta_{n}} +O(1)    \right],
		\end{align}
		where $\mathbb{E}[x_t|x_t\leq x^-]$ is the expectation of $x_t$ on the condition that $x_t\leq x^-$, $u_m=\sup_{t,n} u(t,n) $, and $\delta_{n}=(\mu_n-\mu^*)/u_m$.
	\end{theorem}
	\begin{proof}
		See Appendix \ref{a2}.
	\end{proof}
	
	According to Theorem \ref{ada}, the time order of the learning regret is $O(\ln T)$, indicating that under time-varying input data volume, the TaV is still able to learn which SeV performs the best, and achieves a sublinear deviation compared to the genie-aided optimal policy.
	
	Recall that compared to the existing UCB based algorithms, the major modification under varying input is the introduction of normalized input $\tilde{x}_t$, which dynamically adjusts the weight of exploration and exploitation. 
	As shown in \eqref{ada2}, the consideration of $\tilde{x}_t$ brings an coefficient $\mathbb{E}[x_t|x_t\leq x^-]$ to the learning regret.
	When the input data size is fixed to $x_0$, the coefficient of the learning regret of conventional UCB algorithms is $x_0$. Therefore, by properly selecting the lower threshold $x^-$, we have $\mathbb{E}[x_t|x_t\leq x^-]<x_0$. This implies that the proposed ALTO algorithm can take the opportunity to explore when $x_t$ is small, and achieve lower learning regret. 
	
	Moreover, when the task offloading scenario is simplified to the case with fixed candidate SeVs and identical input data size, the proposed ALTO algorithm reduces to a conventional UCB algorithm, and the lower bound of the learning regret has been investigated in \cite{salomon2011, bubeck2010,bubeck2012}, which is provided in Appendix \ref{lowerbound}. Specifically, the regret lower bound of conventional UCB algorithms is $x_0 u_m \sum_{n\neq a^*}\frac{\delta_n \ln T}{D(n,a^*)}$, where $D(n,a^*)$ is the Kullback-Leibler divergence of the bit offloading delay distributions. Therefore, in the case with varying input, the regret upper bound of ALTO is even possible to be smaller than the lower bound of conventional UCB algorithms, due to the input-awareness.

	\subsection{Joint Consideration of Occurrence-awareness and Input-awareness}
    Finally, we analyze the learning regret by jointly considering the occurrence of vehicles and the variations of input data size. 
    Although these two factors are independent with each other, they actually couple together in the utility function \eqref{ada_utility}, and collectively balance the exploration and exploitation in the learning process. Therefore, it is quite difficult to derive the upper bound of the learning regret in this case.
    
    We study a special case with periodic input and fixed bit offloading delay, and derive the theoretical upper bound to provide some insights. 
	To be specific, assume that the input data size $x_t=\epsilon_0$ when $t$ is even, and $x_t=1- \epsilon_1$ when $t$ is odd, where $\epsilon_0,\epsilon_1\in [0,0.5)$. 
	Let $x^+=1$, and $x^-=\epsilon_0$, thus $\tilde{x}_t=0$ when $x_t=\epsilon_0$, and $\tilde{x}_t=1-\frac{\epsilon_1}{1-\epsilon_0}$ when $x_t=1-\epsilon_1$
	Consider two SeVs appear at $t_1$ and $t_2$ respectively, and $t_1\neq t_2$. 
	Then there are $2$ epochs during $T$ time periods, and we only need to focus on the second epoch, since the first epoch only has one SeV available.
	The bit offloading delay of each SeV is fixed, with $u(t,n)=\mu_n$ for $\forall t, n=1,2$, but unknown in prior.
	Without loss of generality, let $\mu_1\leq\mu_2$, and $\Delta=(\mu_2-\mu_1)/\mu_2$.
	
	The learning regret can be written as
	\begin{align}
		R_T&=\mathbb{E}\left[\sum_{\max\{t_1,t_2\}}^{T}(u(t,n)-\mu_1)\right] \nonumber\\
		&=(\mu_2-\mu_1)\mathbb{E}\left[k^{(2)}_{T,2}\right],
	\end{align}
	where $k^{(2)}_{T,2}$ represents how many times SeV 2 is selected in the second epoch.

    The upper bound for learning regret of ALTO algorithm under periodic input and fixed bit offloading delay is given in the following theorem.
    
    \begin{theorem}\label{t3}
    	Let $\beta_0=2$. With periodic input data size and fixed bit offloading delay, we have:
%
%
    	\begin{align}\label{t3_regret}
    		R_T \leq \frac{2\mu_2\epsilon_0\ln T}{\Delta}+O(1).
    	\end{align}
    \end{theorem}
   	\begin{proof}
   		See Appendix \ref{a4}.
   	\end{proof}
   	
   	The learning regret in \eqref{t3_regret} indicates that, when jointly considering the time-varying feature of input data size and candidate SeV set, the proposed ALTO algorithm still achieves $O(\ln T)$ regret, and focuses on the exploration only when the input is low ($x_t=\epsilon_0$).
   	
   	\begin{conjecture}
   		The proposed ALTO algorithm with random continuous input data size and dynamic SeV set achieves $O(B\ln T)$ learning regret.
   	\end{conjecture}
   	
   	The conjecture follows the insight that,  when the candidate SeV set is identical over time, the learning regret can be derived in a general case with random continuous input and random bit offloading delay, as shown in \eqref{ada2}.
   	When the occurrence time of each SeV is different, within single epoch, the learning regret in \eqref{t3_regret} resembles \eqref{ada2}, both governed by the time order $O(\ln T)$.
   	Following the similar generalization method in \cite{wu2017adaptive}, we may draw a similar conclusion that with random continuous input data size and dynamic SeV set, the learning regret within an epoch is $O(\ln T)$, and the total learning regret is $O(B\ln T)$.

\section{Simulations} \label{sim}
	To evaluate the average delay performance and learning regret of the proposed ALTO algorithm, we carry out simulations in this section.
	We start from a synthetic scenario to evaluate the impact of key parameters, and then simulate a realistic highway scenario using system level simulator Veins\footnote{http://veins.car2x.org/} (VEhicles in Network Simulations) to further verify the proposed ALTO algorithm.
	
	\subsection{Simulation under Synthetic Scenario} 
	We carry out simulations in the synthetic scenario using MATLAB.
	Consider one TaV of interest, with 8 SeVs that appear as candidates during $T=3000$ time periods. 
	The communication range is set to $200 \mathrm{m}$. The distance of the TaV and each candidate SeV ranges within $[10, 200]\mathrm{m}$, and changes randomly from $-10\mathrm{m}$ to $10\mathrm{m}$ in each time period.
	The occurrence and disappearance time of SeVs, as well as their maximum CPU frequency $F_n$ are shown in Table \ref{sev_table}. There are 3 epochs, and each lasts $1000$ time periods. In the first epoch, there are 5 candidate SeVs. At the beginning of the second epoch, a less powerful SeV 5 disappears and SeVs 6 and 7 with higher computing capability appear.
	At the beginning of the third epoch, SeVs 1 and 6 disappear, while SeV 8 with suboptimal computing capability arrives.
	Note that the occurrence and disappearance time of SeVs are unknown to the TaV in prior.
	
  \begin{table*}[!htb]
     	\caption{Candidate SeVs and Maximum CPU Frequency}
     	\label{sev_table}
     	\centering
     	\begin{tabular}{||c||c|c|c|c|c|c|c|c||}
     		\hline
     		Index of SeV & 1 & 2 & 3&4&5&6&7&8\\
     		\hline
     		$F_n$ (GHz)& 3.5&4.5&5&5.5&3&6.5&6&4\\
     		\hline
     		Epoch 1 (time 1$\sim$1000) &${\surd}$ & ${\surd}$ & ${\surd}$& ${\surd}$& ${\surd}$& -- & --&-- \\
     		\hline
     		Epoch 2 (time 1001$\sim$2000)&${\surd}$& ${\surd}$ &  ${\surd}$  &${\surd}$ &$\times$& ${\surd}$& ${\surd}$& -- \\
     		\hline
     		Epoch 3 (time 2001$\sim$3000)& $\times$& ${\surd}$&${\surd}$&${\surd}$& $\times$ & $\times$& ${\surd}$  &${\surd}$ \\
     		\hline
     	\end{tabular}
  \end{table*}
  
  The input data size $x_t$ follows uniform distribution within $[0.2, 1]\mathrm{Mbits}$. The computation intensity is set to $\omega_0=1000 \mathrm{Cycles/bit}$, and the upper and lower thresholds are selected such that $\mathbb{P}\{x\leq x^-\}=0.05$ and $x^+=x^-$. Recall that for each SeV, the allocated CPU frequency $f_{t,n}$ to the TaV is a fraction of the maximum CPU frequency, which is randomly distributed from $20\%F_n$ to $50\%F_n$.	
  The wireless channel state is modeled by an inverse power law $h^{(u)}_{t,n}=h^{(d)}_{t,n}=A_0l^{-2}$, with $A_0=-17.8\mathrm{dB}$, and $l$ is the distance between TaV and SeV \cite{abdulla2016vehicle}.
  Other default parameters include: transmission power $P=0.1\mathrm{W}$, channel bandwidth $W=10\mathrm{MHz}$, noise power $\sigma^2=10^{-13}\mathrm{W}$, and weight factor $\beta_0=0.5$.

	 \begin{figure}[!t]
	 	\centering	
	 	\subfigure[Learning regret.]{\label{syn_compare_regret}			
	 		\includegraphics[width=0.45\textwidth]{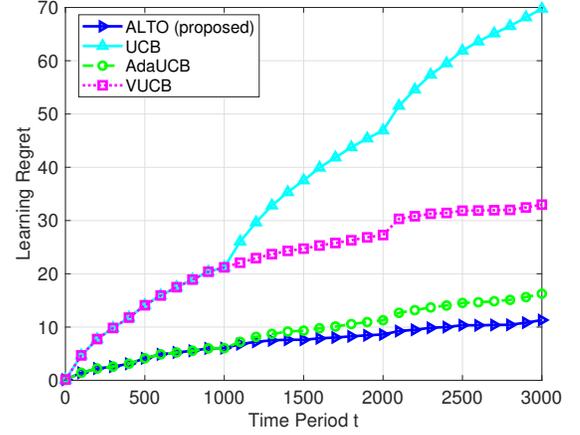}}		
	 	\subfigure[Average delay.]{\label{syn_compare_delay}	
	 		\includegraphics[width=0.45\textwidth]{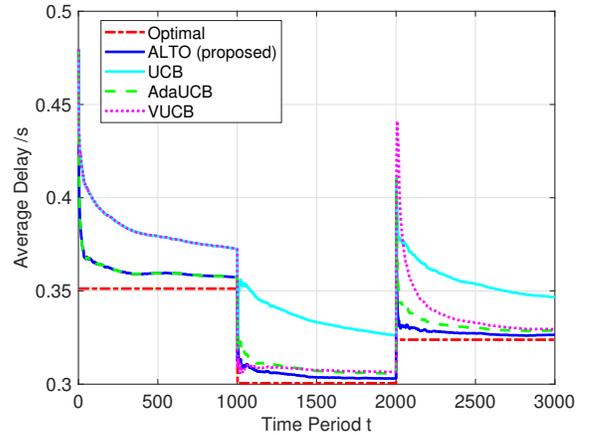}}
    	\caption{Comparison of ALTO algorithm and existing learning algorithms in terms of the learning regret and average delay.}
	 	\label{Syn compare}
	 \end{figure}
	 
	 In Fig. \ref{Syn compare}, the proposed ALTO algorithm is compared with three existing learning algorithms under the MAB framework. 1) \textbf{UCB} is proposed in \cite{auer2002finite}, which is neither input-aware nor occurrence-aware, with padding function $\sqrt{\frac{\beta\ln t}{k_{t-1,n}}}$. 2) \textbf{VUCB} is aware of the occurrence of SeVs, with padding function $\sqrt{\frac{\beta\ln (t-t_n)}{k_{t-1,n}}}$ \cite{bnaya2013social}.  
	 3) \textbf{AdaUCB} is input-aware, with padding function $\sqrt{\frac{\beta(1-\tilde{x}_t)\ln t}{k_{t-1,n}}}$ \cite{wu2017adaptive}. 
	 Note that in the first epoch, VUCB is equivalent to UCB, and AdaUCB is equivalent to ALTO.
	 Besides, in the \textbf{Optimal} genie-aided policy, the TaV always connects to the SeV with minimum expected delay, which is the delay lower bound of the learning algorithm.
	 
	 The comparison of learning regret is shown in Fig. \ref{syn_compare_regret}, which provides two major observations as follows. First, the proposed ALTO algorithm performs the best among the four learning algorithms. To be specific, both VUCB and AdaUCB achieve lower learning regret compared with UCB algorithm, which means that either input-awareness or occurrence-awareness brings adaptivity to the dynamic VEC environment and reduces loss of delay performance through learning. 
	 The joint consideration of these two factors further optimizes the exploration-exploitation tradeoff, and decreases the learning regret by $85\%$, $65\%$ and $30\%$ from that of UCB, VUCB and AdaUCB respectively.
	 Second, the learning regret of ALTO grows sublinearly with time $t$, indicating that the TaV can asymptotically converge to the SeV with optimal delay performance. 
	 As shown in Fig. \ref{syn_compare_delay}, during each epoch, the average delay of ALTO converges faster to the optimal delay than other learning algorithms, and achieves close-to-optimal delay performance.

	\begin{figure}[!t]
		\centering	
			\includegraphics[width=0.45\textwidth]{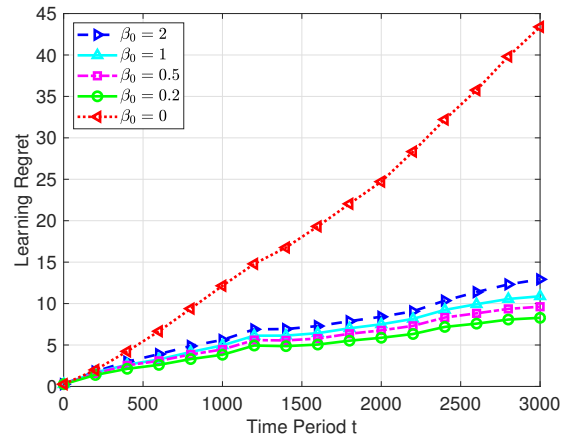}	
			\caption{Learning regret of ALTO under different weight factors $\beta_0$.}
		\label{syn_beta}
	\end{figure} 
    We then consider a single epoch and set SeVs 2-7 in Table \ref{sev_table} as candidates for $3000$ time periods.    
    Fig. \ref{syn_beta} evaluates the impact of weight factor $\beta_0$ on the learning regret.
    When $\beta_0=0$, there is no exploration in the learning process, and the learning regret is drastically worse than those of $\beta_0>0$, since ALTO may stick to a suboptimal SeV for a long time.    
    When $\beta_0>0$, the learning regret grows up slightly as $\beta_0$ increases. 
    Although the existing effort shows that the sublinear learning regret is achieved when $\beta_0>0.5$ \cite{bubeck2010}, in our simulation, the learning regret is lower when $\beta_0=0.2$.
    The reason may be that only a small number of explorations can help the TaV to find the optimal SeV under our settings.

	\begin{figure}[!t]
		\centering	
		\includegraphics[width=0.45\textwidth]{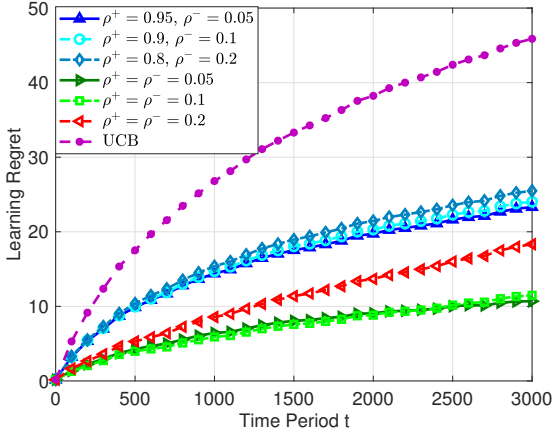}	
		\caption{Learning regret of ALTO under different normalized factors $x^+$ and $x^-$, with $\mathbb{P}\{x\leq x^+\}=\rho^+$ and $\mathbb{P}\{x\leq x^-\}=\rho^-$. }
		\label{syn_norm}
	\end{figure} 

    Finally, we try different pairs of upper and lower thresholds for normalizing the input data size, and evaluate the effect on the learning regret. Define $\mathbb{P}\{x\leq x^+\}=\rho^+$ and $\mathbb{P}\{x\leq x^-\}=\rho^-$, as the probability that the input data size is higher (or lower) than the upper (or lower) threshold.
    Two kinds of thresholds are selected: 1) $\rho^+=\rho^-$, indicating that $ x^+=x^-$ and explorations happen only when $x\leq x^-$. 2) $1-\rho^+=\rho^-$, where explorations also happen when the input data size is between $x^-$ and $ x^+$. 
	As shown in Fig. \ref{syn_norm}, the proposed ALTO algorithm always outperforms UCB algorithm. Moreover, the learning regret under $\rho^+=\rho^-$ is lower than the case when $1-\rho^+=\rho^-$, and achieves the lowest when $\rho^+=\rho^-=0.05$ under our settings, which we set as default.
	
	\subsection{Simulation under Realistic Highway Scenario} 
	In this subsection, simulations are further carried out using system level simulator \emph{Veins}, in order to evaluate the average delay of ALTO under a realistic highway scenario.
	
	The simulation platform Veins integrates a traffic simulator \emph{Simulation of Urban MObility (SUMO)}\footnote{http://www.sumo.dlr.de/userdoc/SUMO.html} and a network simulator \emph{OMNeT++}\footnote{https://www.omnetpp.org/documentation}, and enables to use real maps from \emph{Open Street Map (OSM)}\footnote{http://www.openstreetmap.org/}.
	Vehicular communication protocols including IEEE 802.11p for PHY layer and IEEE 1609.4 for MAC layer are supported by Veins, together with a two-ray interference model \cite{sommer2012} which captures the feature of vehicular channel better.
	
	\begin{figure}[!t]
		\centering	
		\includegraphics[width=0.45\textwidth]{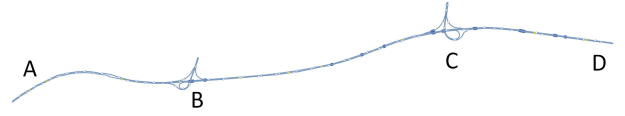}	
		\caption{The highway map used in Veins.}
		\label{map}
	\end{figure} 

	\begin{figure}[!t]
	\centering	
	\subfigure[The arrival probability of SeVs from A to D is $p_{\mathrm{AD}}=0.1$.]{\label{one_d1}			
		\includegraphics[width=0.45\textwidth]{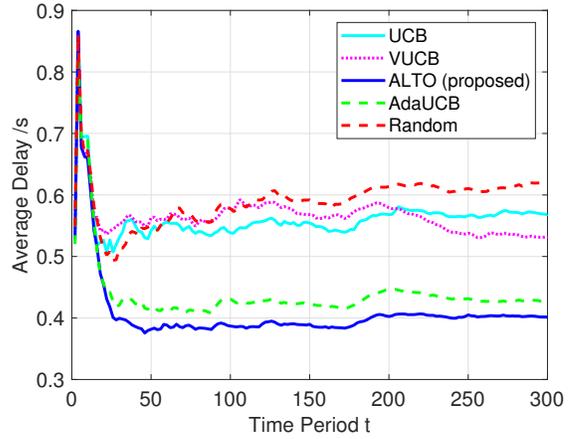}}		
	\subfigure[The arrival probability of SeVs from A to D is $p_{\mathrm{AD}}=0.2$.]{\label{one_d2}	
		\includegraphics[width=0.45\textwidth]{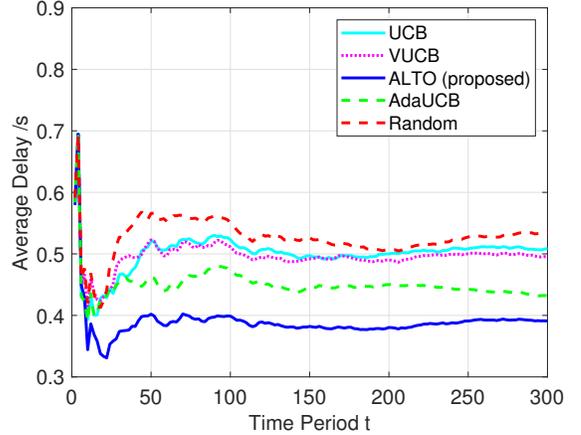}}
	\caption{The average delay performance of ALTO algorithm in the highway scenario with $1$ TaV.}
	\label{delay_one}
\end{figure}

	\begin{figure}[!t]
	\centering	
	\subfigure[The arrival probability of SeVs from A to D is $p_{\mathrm{AD}}=0.1$.]{\label{flow_d1}			
		\includegraphics[width=0.45\textwidth]{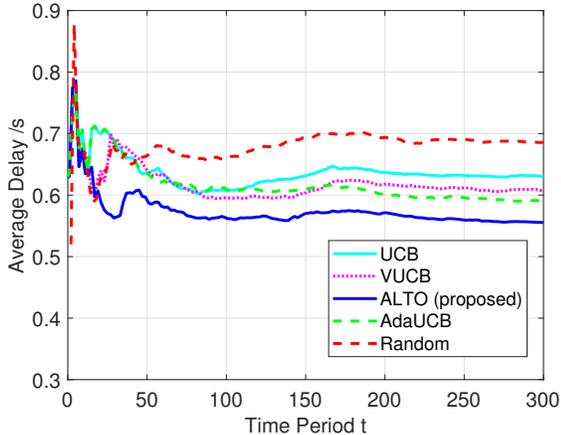}}		
	\subfigure[The arrival probability of SeVs from A to D is $p_{\mathrm{AD}}=0.2$.]{\label{flow_d2}	
		\includegraphics[width=0.45\textwidth]{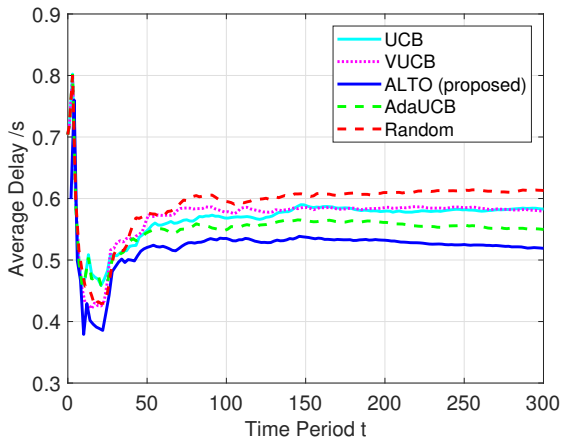}}
	\caption{The average delay performance of ALTO algorithm in the highway scenario with $10$ TaVs, whose inter-arrival time is fixed to $10\mathrm{s}$.}
	\label{delay_flow}
\end{figure}

	A $12\mathrm{km}$ segment of G6 Highway in Beijing is downloaded from OSM and used in our simulation, with two lanes and two ramps, as shown in Fig. \ref{map}. The maximum speed of TaVs and SeVs is set to $60\mathrm{km/h}$.
	The TaV moves from A to D, and SeVs have three routes: A to D, A to C and B to D.
	The arrival of SeVs is modeled by Bernoulli distribution, with probability $p_{\mathrm{AC}}= p_{\mathrm{BD}}=0.05$, and $p_{\mathrm{AD}}$ ranging from $0.1$ to $0.2$ (e.g., $p_{\mathrm{AC}}$ is the probability of the generation of a SeV which departs at A and leaves the road from C at each second).
	Besides the aforementioned UCB, VUCB and AdaUCB algorithms, we also adopt a naive \textbf{Random} policy as a baseline, where the TaV randomly selects a SeV for task offloading in each time period.

	Fig. \ref{delay_one} shows the average delay performance with a single TaV, which means the density of SeV is much higher than that of TaV. 
	And in Fig. \ref{delay_flow}, we consider 10 TaVs that depart every 10 seconds. In this case, each TaV is within some other TaVs' communication range, and thus they might compete for bandwidth and computing resources.
	We make three major observations as follows.
	First, the proposed ALTO algorithm always outperforms the other learning algorithms and the random policy, illustrating that ALTO can adapt to the vehicular environment better. 
	To be specific, compared with the UCB algorithm, when $p_{\mathrm{AD}}=0.1$, ALTO can reduce the average delay by about $30\%$ under single TaV case (Fig. \ref{one_d1}), and  $13\%$ under multi-TaV scenario (Fig. \ref{flow_d1}).
	Second, the average delay grows up when the density of TaV becomes high, since each SeV may serve multiple TaVs simultaneously.
	Besides, as shown in Fig. \ref{delay_flow}, when the density of TaV is high, the average delay performance decreases as the arrival probability of  SeV increases, since the computing resources are more sufficient.

	\section{Conclusions} \label{con}
	In this paper, we have studied the task offloading problem in vehicular edge computing (VEC) systems, and proposed an adaptive learning-based task offloading (ALTO) algorithm to minimize the average offloading delay.
	The proposed algorithm enables each task vehicle (TaV) to learn the delay performance of service vehicles (SeVs) in a distributed manner, without frequent exchange of state information.
	Considering the time-varying features of task workloads and candidate SeVs, we have modified the existing multi-armed bandit (MAB) algorithms to be input-aware and occurrence-aware, so that ALTO algorithm is able to adapt to the dynamic vehicular task offloading environment.
	Theoretical analysis has been carried out, providing a sublinear learning regret of the proposed algorithm.
	We have evaluated the average delay and learning regret of ALTO under a synthetic scenario and a realistic highway scenario, and shown that the proposed algorithm can achieve low delay performance, 
	and decrease the learning regret up to $85\%$ and the average delay up to $30\%$, compared with the classical upper confidence bound algorithm.
	
	As future work, we plan to formulate the task offloading problem based on adversarial MAB framework \cite{bubeck2012}, where no stochastic assumptions are made on the delay performance of SeVs. The adversarial setting makes learning more difficult, but may perform better under more complicated vehicular environments such as urban scenarios. Besides, we plan to consider the joint resource allocation of vehicles and infrastractures in the VEC system, in order to further optimize the delay performance.
	

\appendices{}

\section{Proof of Lemma 1} \label{a1}

In the $b$th epoch, the learning regret is
\begin{align}  \label{a1_1}
R_{b}&=x_0  \mathbb{E}\left[\sum_{t=t_b}^{t'_b}u(t,n)-\mu_b^*\right] \nonumber\\
&=x_0 \mathbb{E}\left[\sum_{n\in \mathcal{N}_b} k_{n,b}u_m\delta_{n,b} \right]\nonumber\\
&=x_0u_m\sum_{n\neq a_b^*} \delta_{n,b} \mathbb{E}[k_{n,b}],
\end{align}
where $k_{n,b}$ is the number of tasks offloaded to SeV $n\in \mathcal{N}_b$ in the $b$th epoch.
According to Lemma 1 in \cite{bnaya2013social} and Theorem 1 in \cite{auer2002finite}, when $\beta_0=2$, the expected number of tasks offloaded to a suboptimal SeV has an upper bound as follows
\begin{align} \label{a1_2}
\mathbb{E}[k_{n,b}] &\leq \frac{8\ln (t'_b-t_n)}{\delta^2_{n,b} }+1+\frac{\pi^2}{3}.
\end{align}

Substituting \eqref{a1_2} into \eqref{a1_1}, we get:
\begin{align}
&R_{b} =x_0u_m\sum_{n\neq a_b^*} \delta_{n,b} \mathbb{E}[k_{n,b}] \nonumber\\
&\leq x_0u_m \left[\sum_{n\neq a_b^*}\frac{8\ln (t'_b-t_n)}{\delta_{n,b}} + \left(1+ \frac{\pi^2}{3}\right)\sum_{n\neq a_b^*}\delta_{n,b} \right].
\end{align}

Thus we can prove Lemma \ref{vucb_epoch}.

\section{Proof of Theorem 1} \label{a3}
We have $t'_b\leq T$ for $\forall b=1,2,...,B$. Following Lemma \ref{vucb_epoch}, the learning regret in the $b$th epoch can be bounded from above as: 
\begin{align}
R_{b}& \leq x_0u_m \left[\sum_{n\neq a_b^*}\frac{8\ln (t'_b-t_n)}{\delta_{n,b}} + \left(1+ \frac{\pi^2}{3}\right)\sum_{n\neq a_b^*}\delta_{n,b} \right]\nonumber\\
& \leq  x_0u_m \left[\sum_{n\neq a_b^*}\frac{8\ln T}{\delta_{n,b}} +O(1)\right].
\end{align}

By summing over the learning regrets of the $B$ epochs, we have:
\begin{align}
R_{T} =\sum_{b=1}^{B}R_{b}\leq  x_0u_m \sum_{b=1}^{B}\left[\sum_{n\neq a_b^*}\frac{8\ln T}{\delta_{n,b}} +O(1)\right].
\end{align}

Thus Theorem \ref{vucb_T} is proved.

\section{Proof of Theorem 2} \label{a2}
When $\beta_0=2$ and $B=1$, the utility function in \eqref{ada_utility} is
\begin{align} 
\hat{u}_{t,n}=\bar{u}_{t-1,n}-u_m\sqrt{\frac{2(1-\tilde{x}_t)\ln t}{k_{t-1,n}}}.
\end{align}	

The decision making function in \eqref{ada_obj} can be written as
\begin{align} \label{a2_1}
a_t&=\arg\min_{n \in \mathcal{N}_1} \hat{u}_{t,n} \nonumber\\
&=\arg\min_{n \in \mathcal{N}_1} \left\{  \bar{u}_{t-1,n}-u_m\sqrt{\frac{2(1-\tilde{x}_t)\ln t}{k_{t-1,n}}} \right\}\nonumber\\
&=\arg\min_{n \in \mathcal{N}_1}\left\{ \frac{\bar{u}_{t-1,n}}{u_m}-\sqrt{\frac{2(1-\tilde{x}_t)\ln t}{k_{t-1,n}}} \right\} \nonumber\\
&=\arg\max_{n \in \mathcal{N}_1}\left\{1- \frac{\bar{u}_{t-1,n}}{u_m}+\sqrt{\frac{2(1-\tilde{x}_t)\ln t}{k_{t-1,n}}} \right\}.
\end{align}  

The learning regret can be written as
\begin{align} \label{a2_2}
&~R_T =\mathbb{E}\left[\sum_{t=1}^{T}x_t(u(t,n)-\mu^*)\right]    \nonumber\\
&=u_m\mathbb{E}\left[\sum_{t=1}^{T}x_t\left\{\left(1-\frac{\mu^*}{u_m}\right)-\left(1-\frac{u(t,n)}{u_m}\right)\right\}\right].
\end{align}

Since $1- \frac{\bar{u}_{t-1,n}}{u_m} \in [0,1]$, and $1-\frac{u(t,n)}{u_m} \in [0,1]$, 
the task offloading problem can be transformed to the opportunistic bandit problem defined in Section \ref{sys} in our previous work \cite{wu2017adaptive}, with equivalent definitions of learning regret, utility and decision making (as shown in \cite{wu2017adaptive}, eq. (1-3)).
By leveraging Lemma 7 and Appendix C.2 in \cite{wu2017adaptive}, we can get the upper bound of $\mathbb{E}[k_{T,n}]$, as shown in Theorem \ref{ada}(1).
By leveraging Theorem 3 and Appendix C.2 in \cite{wu2017adaptive}, we can get the upper bound of the learning regret $R_T$, as shown in Theorem \ref{ada}(2).

\section{Regret Lower Bound} \label{lowerbound}

The regret lower bound of classical UCB algorithms has been investigated in \cite{salomon2011, bubeck2010,bubeck2012}. In the following, we provide a regret lower bound of ALTO in a simple task offloading case, with identical input data size $x_0$ and fixed candidate set of SeVs $\mathcal{N}$ (and thus the index of epoch $b$ is omitted).

\begin{lemma} \label{lowerbounducb}
	When the candidate SeV set is not time-varying, and the input data size is identical over time, the learning regret can be bounded from above as:
	\begin{align}
		R_T \geq x_0 u_m \sum_{n\neq a^*}\frac{\delta_n \ln T}{D(n,a^*)},
	\end{align}
	where $D(n,a^*)$ is the Kullback-Leibler divergence of the bit offloading delay distributions of SeV $n$ and optimal SeV $a^*$.
\end{lemma}
\begin{proof}
	With fixed SeV set and identical input data size, the proposed ALTO algorithm reduces to the classical UCB algorithm. According to \cite{salomon2011}, Theorem 5, when $T\rightarrow +\infty$, the number of tasks offloaded to a suboptimal SeV $n$ can be bounded as follows
	\begin{align} \label{lowerboundtime}
		\mathbb{E}[k_{T,n}]\geq \frac{\ln T}{D(n,a^*)}.
	\end{align}	
	Substituting \eqref{lowerboundtime} into \eqref{a1_1}, the learning regret $R_T$ can be bounded as
	\begin{align}
		R_T=&x_0u_m\sum_{n\neq a^*} \delta_{n} \mathbb{E}[k_{T,n}] 	\geq x_0 u_m \sum_{n\neq a^*}\frac{\delta_n \ln T}{D(n,a^*)}.	
	\end{align}
\end{proof}

\section{Proof of Theorem 3} \label{a4}
The proof of Theorem \ref{t3} follows the similar idea in \cite{wu2017adaptive}, while the major difference is that the two SeVs appear at $t_1$ and $t_2$ respectively.
Let $t_0=\max\{t_1,t_2\}$. We only needs to bound the learning regret in the second epoch, from time $t_0$ to time $T$. 

We first bound the number of tasks offloaded to the suboptimal SeV.

\begin{lemma} \label{t3_lemma1}
	With periodic input of tasks and fixed bit offloading delay of SeVs, 
	\begin{align} \label{t3_l1}
		k^{(2)}_{t,2}\leq \frac{\beta_0 \ln t}{\Delta^2}+1.
	\end{align}
\end{lemma}
\begin{proof}
	 First, \eqref{t3_l1} holds for $t=t_0$ and $t_0+1$.
	 For $t_0\geq t_0+2$, we prove the lemma by contradiction. For simplicity, we use $k_{t,2}$ rather than $k^{(2)}_{t,2}$.
	 If \eqref{t3_l1} does not hold, there exists at least one $\tau\geq t_0+2$, such that 
	 \begin{align}
	 	k_{\tau-1,2}&\leq \frac{\beta_0 \ln (\tau-1)}{\Delta^2}+1, \\
	 	k_{\tau,2}&> \frac{\beta_0 \ln \tau}{\Delta^2}+1.
	 \end{align}
	 Since $\ln \tau>\ln (\tau-1)$, SeV 2 is selected at time $\tau$.
	 
	 According to the utility function in \eqref{ada_obj}, when $x_t=\epsilon_0$,
	 \begin{align}
	 	\mu_1-\sqrt{\frac{\beta\ln (\tau-t_1)}{k_{\tau-1,1}}} \geq \mu_2-\sqrt{\frac{\beta\ln (\tau-t_2)}{k_{\tau-1,2}}}.
	 \end{align}
	 Thus $\Delta=\frac{\mu_2-\mu_1}{\mu_2}<\frac{1}{\mu_2}\sqrt{\frac{\beta\ln (\tau-t_2)}{k_{\tau-1,2}}}\leq \sqrt{\frac{\beta_0\ln \tau}{k_{\tau-1,2}}}$, and $k_{\tau-1,2}<\frac{\beta_0 \ln \tau}{\Delta^2}$.
	 Then $k_{\tau,2}\leq k_{\tau-1,2}+1<\frac{\beta_0 \ln \tau}{\Delta^2}+1$.
	 
	 Similar proof can be carried out when $x_t=1-\epsilon_1$. Thus we prove Lemma \ref{t3_lemma1}.	 
\end{proof}

Then we prove that the proposed ALTO algorithm can explore sufficiently, such that when the input data size is large, it always selects the optimal SeV 1.
\begin{lemma}\label{t3_lemma2}
	With periodic input of tasks and fixed bit offloading delay of SeVs, there exists $T_1$, such that $a_t=1$ when $t\geq T_1$ and $x_t=1-\epsilon_1$. 
\end{lemma}
\begin{proof}
	First, define an auxiliary function
	\begin{align}
		h(t)=\frac{\beta_0 \ln (2t-t_2)}{\Delta^2}\left(  1+ \sqrt{\frac{2\beta_0\ln 2t}{\Delta^2(2t-1-t_0)}}  \right)^{-2},
	\end{align}
	 and $f(t)=\int_{t_0}^{t} \min(h'(s),1)\mathrm{d}s+h(t_0)$.
	We prove that $k_{2t,2}\geq f(t)$. It is easy to see that $k_{2t,2}\geq f(t)$ holds when $t=t_0$ and $t_0+1$.
	Assume that there exists $\tau\geq t_0+2$, such that $k_{2(\tau-1),2}\geq f(\tau-1)$, but $k_{2\tau,2}< f(\tau)$.
	Since $f(\tau)-f(\tau-1)=\int_{\tau-1}^{\tau} \min(h'(s),1)\mathrm{d}s\leq 1$, and $k_{2(\tau-1),2}$, $k_{2\tau-1,2}$, $k_{2\tau,2}$ are integers, we have $k_{2(\tau-1),2}=k_{2\tau-1,2}=k_{2\tau,2}$.
	Thus SeV 1 is selected at time $2\tau$.
	
	When $t=2\tau$, $x_t=\epsilon_0$. According to the utility function in \eqref{ada_obj}, we have
	\begin{align}
	\mu_1-\sqrt{\frac{\beta\ln (2\tau-t_1)}{k_{2\tau-1,1}}} \leq  \mu_2-\sqrt{\frac{\beta\ln (2\tau-t_2)}{k_{2\tau-1,2}}}.
	\end{align}
	Thus
	\begin{align}
		\Delta=\frac{\mu_2-\mu_1}{\mu_2}\geq \sqrt{\frac{\beta_0\ln (2\tau-t_2)}{k_{2\tau-1,2}}}- \sqrt{\frac{\beta_0\ln (2\tau-t_1)}{k_{2\tau-1,1}}}.
	\end{align}
	When $\tau$ is sufficiently large, $k_{2\tau-1,1}\geq (2\tau-1-t_0)/2$.
	Then
	\begin{align}
	\Delta=\frac{\mu_2-\mu_1}{\mu_2}\geq \sqrt{\frac{\beta_0\ln (2\tau-t_2)}{k_{2\tau-1,2}}}- \sqrt{\frac{2\beta_0\ln (2\tau-t_1)}{2\tau-1-t_0}}.
	\end{align}
	And thus $k_{2\tau,2}=k_{2\tau-1,2}\geq h(\tau)\geq f(\tau)$, which contradicts the assumption.
	
	Therefore, $k_{2t,2}\geq f(t)$ holds for any $t\geq t_0$.
	
	When $x_t=1-\epsilon_1$, $t$ is odd. Let $t=2\tau+1$, the utility function of SeV 2 is
	\begin{align}
		\hat{u}_{t,2}&=\bar{u}_{t-1,2}- \sqrt{\frac{\beta \epsilon_1\ln (2\tau+1-t_2)}{(1-\epsilon_0)k_{2\tau,2}}} \nonumber\\
		&\geq \mu_2-\sqrt{\frac{\beta \epsilon_1\ln (2\tau+1-t_2)}{(1-\epsilon_0)f(\tau)} }
	\end{align}
	Note that $\frac{1-\epsilon_0}{\epsilon_1}>1$.
	There exists $T_1$, such that when $t\geq T_1$, $\frac{\ln (2\tau+1-t_2)}{f(\tau)}<\frac{\Delta^2}{\beta_0}\frac{1-\epsilon_0}{\epsilon_1}$.
	Therefore,
	\begin{align}
		\hat{u}_{t,2}&\geq \mu_2-\sqrt{\frac{\beta \epsilon_1\ln (2\tau+1-t_2)}{(1-\epsilon_0)f(\tau)} }\nonumber\\
		&>\mu_2-\sqrt{ \frac{\beta \epsilon_1}{(1-\epsilon_0) } \frac{\Delta^2}{\beta_0}\frac{1-\epsilon_0}{\epsilon_1} } \nonumber\\
		&=\mu_2-\mu_2\Delta=\mu_1>\hat{u}_{t,1},
	\end{align}
	which indicates that SeV 1 is selected. Thus Lemma \ref{t3_lemma2} is proved.	 
\end{proof}

Finally, by letting $\beta_0=2$ and combining Lemma \ref{t3_lemma1} and Lemma \ref{t3_lemma2}, Theorem \ref{t3} can be derived.

	%

\begin{thebibliography}{100}
	
\bibitem{Sun2018ICC}
Y.~Sun, X.~Guo, S.~Zhou, Z. Jiang, X.~Liu, and Z.~Niu, ``Learning-based task offloading for vehicular cloud computing systems," in \emph{Proc. IEEE Int. Conf. Commun. (ICC)}, Kansas City, MO, USA, May 2018. 

\bibitem{hu2015mobile}
Y. C. Hu, M. Patel, D. Sabella, N. Sprecher, and V. Young, ``Mobile
edge computing: A key technology towards 5G,'' \emph{ETSI White Paper No.11},
vol. 11, 2015.

\bibitem{shih17}
Y. Y. Shih, W. H. Chung, A. C. Pang, T. C. Chiu, and H. Y. Wei, ``Enabling low-latency applications in fog-radio access networks," \emph{IEEE Netw.}, vol. 31, no. 1, pp. 52-58, Feb. 2017.

\bibitem{mao2017mobile}
Y.~Mao, C.~You, J.~Zhang, K.~Huang, and K.~B. Letaief, ``A survey on mobile
edge computing: The communication perspective,'' \emph{IEEE Commun. Surveys
	Tuts.}, vol. 19, no. 4, pp. 2322-2358, 2017.

\bibitem{mach2017mobile}
P. Mach, and Z. Becvar, ``Mobile edge computing: A survey on architecture and computation offloading,'' \emph{IEEE Commun. Surveys Tut.}, vol. 19, no. 3, pp. 1628-1656, 2017.

\bibitem{yu2018survey}
W. Yu, F. Liang, X. He, W. G. Hatcher, C. Lu, J. Lin, and X. Yang, ``A survey on the edge computing for the Internet of things," \emph{IEEE Access}, vol. 6, pp. 6900-6919, 2018.

\bibitem{you2016energy}
C.~You, K.~Huang, H.~Chae, and B.-H. Kim, ``Energy-efficient resource
allocation for mobile-edge computation offloading,'' \emph{IEEE Trans.
	Wireless Commun.}, vol.~16, no.~3, pp. 1397--1411, Mar. 2016.

\bibitem{chen2016efficient}
X.~Chen, L.~Jiao, W.~Li, and X.~Fu, ``Efficient multi-user computation
offloading for mobile-edge cloud computing,'' \emph{IEEE Trans. Netw.},
vol.~24, no.~5, pp. 2795--2808, Oct. 2016.

\bibitem{jin2018}
A. L. Jin, W. Song, and W. Zhuang, ``Auction-based resource allocation for sharing cloudlets in mobile cloud computing," \emph{IEEE Trans. Emerg. Topics Comput.}, vol. 6, no. 1, pp. 45-57, 2018.

\bibitem{intel}
Intel, ``Self-driving car technology and computing requirements," [Online] Available: 
https://www.intel.com/content/www/us/en/automotive/
driving-safety-advanced-driver-assistance-systems-self-driving-technol
ogy-paper.html

\bibitem{abdel2015vehicle}
S. Abdelhamid, H. Hassanein, and G. Takahara, ``Vehicle as a resource (VaaR),'' \emph{IEEE Netw.}, vol. 29, no. 1, pp. 12-17, Feb. 2015.

\bibitem{bitam2015vanet}
S. Bitam, A. Mellouk, and S. Zeadally, ``VANET-cloud: A generic cloud computing model for vehicular ad hoc networks,'' \emph{IEEE Wireless Commun.}, vol. 22, no. 1, pp. 96-102, Feb. 2015.

\bibitem{choo2017sdvc}
J. S. Choo, M. Kim, S. Pack, and G. Dan, ``The software-defined vehicular cloud: A new level of sharing the road," \emph{IEEE Veh. Technol. Mag.}, vol. 12, no. 2, pp. 78-88, Jun. 2017.

\bibitem{aibo}
X. Cheng, C. Wang, B. Ai, and H. Aggoune, ``Envelope level crossing rate and average fade duration of nonisotropic vehicle-to-vehicle Ricean fading channels." \emph{IEEE Trans. Intell. Transp. Syst.}, vol. 15, no. 1, pp. 62-72, Feb. 2014.

\bibitem{auer2002finite}
P.~Auer, N.~Cesa-Bianchi, and P.~Fischer, ``Finite-time analysis of the
multiarmed bandit problem,'' \emph{Machine learning}, vol.~47, no. 2-3, pp.
235--256, 2002.


\bibitem{chen2011opp}
L. Chen, S. Iellamo, and M. Coupechoux, ``Opportunistic Spectrum Access with Channel Switching Cost for Cognitive Radio Networks,'' in \emph{Proc. IEEE Int. Conf. Commun. (ICC)}, Kyoto, Japan, Jun. 2011.



\bibitem{shen2016non}
C.~Shen, C.~Tekin, and M.~van~der Schaar, ``A non-stochastic learning approach
to energy efficient mobility management,'' \emph{IEEE J. Sel. Areas Commun.},
vol.~34, no.~12, pp. 3854--3868, Dec. 2016.

\bibitem{sun2017emm}
Y. Sun, S. Zhou, and J. Xu, ``EMM: Energy-Aware Mobility Management for Mobile Edge Computing in Ultra Dense Networks,''  \emph{IEEE J. Sel. Areas Commun.}, vol. 35, no. 11, pp. 2637-2646, Nov. 2017.

\bibitem{wu2017adaptive}
H. Wu, X. Guo, and X. Liu, ``Adaptive exploration-exploitation tradeoff for opportunistic bandits.'' in \emph{Proc. International Conference on Machine Learning (ICML)}, Stockholm, Sweden, Jul. 2018.

\bibitem{5gv2x}
3GPP, ``Study on enhancement of 3GPP support for 5G V2X services," 
3GPP TR 22.886, V15.1.0, Mar. 2017, 

\bibitem{zhangshan}
S. Zhang, J. Chen, F. Lyu, N. Cheng, W. Shi, and X. Shen, ``Vehicular communication networks in automated driving era,'' [Online]. Available:
\url{https://arxiv.org/abs/1805.09583}

\bibitem{vehcrowd}
J. Ni, A. Zhang, X. Lin, and X. S. Shen, ``Security, privacy, and fairness in fog-based vehicular crowdsensing," \emph{IEEE Commun. Mag.}, vol. 55, no. 6, pp. 146-152, Jun. 2017.

\bibitem{fengave}
J. Feng, Z. Liu, C. Wu, and Y. Ji, ``AVE: autonomous vehicular edge computing framework with aco-based scheduling,'' \emph{IEEE Trans. Veh. Technol.}, vol. 66, no. 12, pp. 10660-10675, Dec. 2017.

\bibitem{zheng2015smdp}
K. Zheng, H. Meng, P. Chatzimisios, L. Lei, and X. Shen, ``An \protect{SMDP}-based resource allocation in vehicular cloud computing systems,'' \emph{IEEE Trans. Ind. Electron.}, vol. 62, no. 12, pp. 7920-7928, Dec. 2015.

\bibitem{Jiang2017IoT}
Z.~Jiang, S.~Zhou, X.~Guo, and Z.~Niu, ``Task replication for deadline-constrained vehicular cloud computing: Optimal policy, performance analysis and implications on road traffic," \emph{IEEE Internet Things J.}, vol. 5, no. 1, pp. 93-107, Feb. 2018.



\bibitem{kenney2011dsrc}
J. B. Kenney, ``Dedicated short-range communications (DSRC) standards in the United States,'' \emph{Proceedings of the IEEE}, vol. 99, no. 7, pp. 1162-1182, Jul. 2011.



\bibitem{zhang2013energy}
W. Zhang, Y. Wen, K. Guan, D. Kilper, H. Luo, and D. O.Wu, ``Energy-optimal mobile cloud computing under
stochastic wireless channel," \emph{IEEE Trans. Wireless Commun.}, vol. 12,
no. 9, pp. 4569–4581, Sep. 2013.

\bibitem{grundmann2010efficient}
M. Grundmann, V. Kwatra, M. Han, and I. Essa, ``Efficient hierarchical
graph-based video segmentation," in \emph{Proc. IEEE Conf. Comput.
Vis. Pattern Recognit. (CVPR)}, San Francisco, CA, USA, Jun. 2010.

\bibitem{bnaya2013social}
Z.~Bnaya, R.~Puzis, R.~Stern, and A.~Felner, ``Social network search as a
volatile multi-armed bandit problem,'' \emph{HUMAN}, vol.~2, no.~2, pp.~84--98, 2013.


\bibitem{salomon2011}
A. Salomon, J. Y. Audibert, and I. E. Alaoui, ``Regret lower bounds and extended upper confidence bounds policies in stochastic multi-armed bandit problem,'' [Online]. Available: \url{https://arxiv.org/abs/1112.3827}.


\bibitem{bubeck2010}
S. Bubeck, ``Bandits games and clustering
foundations,'' Ph.D. Dissertation, Universite des Sciences
et Technologie de Lille-Lille I, 2010


\bibitem{bubeck2012}
S. Bubeck, and N. Cesa-Bianchi, ``Regret analysis of stochastic and nonstochastic multi-armed bandit problems,'' \emph{Foundations and Trends in Machine Learning}, vol. 5, no.1, pp. 1-122, Dec. 2012. 


\bibitem{abdulla2016vehicle}
M. Abdulla, E. Steinmetz, and H. Wymeersch,``Vehicle-to-vehicle communications with urban intersection path loss models,'' in \emph{Proc. IEEE Global Commun. Conf. (GLOBECOM)}, Washington, DC, USA, Dec. 2016.




\bibitem{sommer2012}
C. Sommer, S. Joerer, and F. Dressler, ``On the Applicability of Two-Ray Path Loss Models for Vehicular Network Simulation," in \emph{Proc. IEEE Veh. Netw. Conf. (VNC)}, Seoul, Korea, Nov. 2012, pp. 64-69.







\end{thebibliography}

\end{document}